\documentclass[a4paper]{amsart}

\usepackage{amssymb,amsmath,graphicx,epstopdf}
\usepackage{enumitem}
\usepackage{framed}
\usepackage[all]{xy}
\usepackage{tikz}
\usepackage{algorithm}
\usepackage{algpseudocode}

\usetikzlibrary{positioning}

\theoremstyle{plain}
\newtheorem{theorem}{Theorem}[section]
\newtheorem{definition}[theorem]{Definition}
\newtheorem{lemma}[theorem]{Lemma}
\newtheorem{corollary}[theorem]{Corollary}

\newtheorem{example}[theorem]{Example}

\newtheorem{proposition}[theorem]{Proposition}

\newtheorem*{openproblem}{Problem}

\newcommand{\CSP}{\operatorname{CSP}}

\begin{document}

\title[New algorithm for Maltsev CSPs]{A new algorithm for constraint satisfaction problems with Maltsev templates}

\author{Dejan Deli\'c}
\address{Department of Mathematics, Ryerson University,  Canada}
\email{ddelic@ryerson.ca}

\author{Aklilu Habte}
\address{Deaprtment of Mathematics, Ryerson University, Canada}
\email{a2habte@ryerson.ca}

\thanks{The first author gratefully acknowledges support by the Natural Sciences and
Engineering Research Council of Canada in the form of a Discovery Grant.}

\begin{abstract} 
In this article, we provide a new algorithm for solving constraint satisfaction problems with Maltsev constraints, based on the new notion of Maltsev consistency.
\end{abstract}

\maketitle

\section*{Introduction}
\noindent One of the fundamental problems in constraint programming and, more widely, in the field of artificial intelligence, is the problem of determining the computational complexity of constraint satisfaction problems (CSPs, for short). The problem, in its full generality, is \textsc{NP}-complete but many of its subclasses are tractable with algorithms which are well understood. In this paper, we adopt the following to studying the complexity of CSPs: we study the restrictions of  the instances by allowing a fixed set of constraint relations. This approach is generally referred to in the literature as a \emph{constraint language} or, a \emph{fixed template}  (\cite{b-j-k}). This point of view has lead to significant progress in the study of the complexity of constraint satisfaction in the past 15 years or so.

The first general result obtained by this approach was Schaefer's dichotomy for Boolean CSPs. Schaefer proved that CSPs arising from constraint languages over 2-element domains are either solvable in polynomial time or \textsc{NP}-complete. The algebraic approach has subsequently yielded a number of important results. Among others, A. Bulatov \cite{bul3} extended
Schaefer's~\cite{sch} result on 2-element domains to prove the CSP  
dichotomy conjecture for 3-element domains. Barto and
Kozik~\cite{b-k2} gave a complete algebraic description of the
constraint languages over finite domains that are solvable by local
consistency methods (these problems are said to be of \emph{bounded
  width}) and as a consequence it is decidable to determine whether a
constraint language can be solved by such methods. 

The proof of the Bounded Width Conjecture by L. Barto and M. Kozik is a typical example of the result which demonstrates the influence universal algebra has exerted over the study of parametrized CSPs. The main problem driving the algebraic approach to the study of parametrized CSPs can be stated in the following way: given a constraint language $\Gamma$, what classes of operations preserving the relations in $\Gamma$ guarantee the existence of a ``nice" algorithm solving the problem? 

One of the early landmark results in this direction were the proof of A. Bulatov,  and a substantially simplified version of the same result by  A. Bulatov and V. Dalmau (\cite{bulatov2006simple}) concerning the existence of such an algorithm for a fairly general, yet rather natural class of parametrized CSPs, those whose parametrizing algebra is \emph{Maltsev}, i.e. in which the constraint language is invariant under an algebraic operation satisfying the condition $m(x,x,y)\approx m(y,x,x)\approx y$, for all elements $x$ and $y$ of the algebra. Their algorithm, also known as the Generalized Gaussian Elimination, provided a common generalization for already known algorithms for solving CSPs over affine domains, CSPs on finite groups with near subgroups, etc. The result of Bulatov and Dalmau is based on the algebraic fact that, given any Maltsev algebra $\mathbb{A}$, any subpower of $\mathbb{A}^n$ has a generating set of polynomial (in fact, linear) size in $n$. This approach was further generalized in \cite{BIMMVW} to show that a modification of Bulatov-Dalmau algorithm solves all CSPs over the so called domains \emph{with few subpowers}, i.e. for all the domains $\mathbb{A}$ with the property that any subpower of $\mathbb{A}^n$ has a generating set of polynomial size in $n$. 

The algorithms presented in \cite{bulatov2006simple} and \cite{BIMMVW} require explicit knowledge of the algebraic operations witnessing the few subpowers property. This, in itself, may be viewed as problematic, since, in practical applications, the CSP is generally presented in the form of its constraint language and computing the required term is a highly nontrivial problem in terms of its complexity. Secondly, those algorithms do not provide ``short" proofs of unsatisfiability in the same way local consistency checks do. Finally, there is increasing evidence that the solvability of a CSP is intimately linked to the solvability of very particular subinstances, namely, the subinstances over finite simple algebras in the variety of the parametrizing algebra, while the Generalized Gaussian Elimination algorithm does not make use of the structure theory of congruence permutable or congruence modular varieties. The algorithm we present in this article is largely driven by the attempts of the authors to address some of these issues and gain better understanding of how the structural theory of Maltsev varieties influences the solvability of Maltsev CSPs.

\section{Preliminaries}

\subsection{Constraint Satisfaction Problem}
The notion of a constraint satisfaction problem provides us with a natural framework for a variety of problems which require simultaneous satisfiability of a number of conditions on a given set of variables. More formally,

\begin{definition} An \emph{instance} of the CSP is a triple $\mathcal{I}=(V,A,\mathcal{C})$, where $V=\{x_1,\ldots,x_n\}$ is a finite set of \emph{variables}, $A$ is a finite domain for the variables in $V$, and $\mathcal{C}$ is a finite set of \emph{constraints} of the form $C=(S,R_S)$, where $S$, the \emph{scope} of the constraint, is a $k$-tuple of variables $(x_{i_1},\ldots,x_{i_k})\in V^k$ and $R_S$ is a $k$-ary relation $R_S\subseteq A^k$, called the \emph{constraint relation} of $C$.

A \emph{solution} for the instance $\mathcal{I}$ is any assignment $f:V\rightarrow A$, such that, for every constraint $C=(S,R_S)$ in $\mathcal{C}$, $f(S)\in R_S$.
\end{definition}

A relational structure $\mathbf{A}=(A, \Gamma)$, defined over the domain $A$ of the instance $\mathcal{I}$, where $\Gamma$ is a finite set of relations on $A$, is often referred to as a \emph{constraint language}, and the relations from $\Gamma$ form the signature of $\mathbf{A}$. An instance of $\CSP (\mathbf{A})$ will be an instance of the CSP such that all constraint relations belong to $\mathbf{A}$. 



\subsection{Basic Algebraic Concepts}\label{sec:Algebra}
\noindent In this subsection, we introduce concepts from universal algebra which will be used in the remainder of the paper. Two good refrences for a more in-depth overview of universal algebra are \cite{Burris1981} and \cite{bergman}.

An \emph{algebra} is an ordered pair $\mathbb{A}=(A, F)$, where $A$ is a nonempty set, the \emph{universe} of $\mathbb{A}$, while $F$ is the set of \emph{basic operations} of $\mathbb{A}$, consisting of functions of arbitrary, but finite, arities on $A$. The list of function symbols and their arities is the \emph{signature} of $\mathbb{A}$. 

A \emph{subuniverse} of the algebra $\mathbb{A}$ is a nonempty subset $B\subseteq A$ closed under all operations of $\mathbb{A}$. If $B$ is a subuniverse of $\mathbb{A}$, by restricting all operations of $\mathbb{A}$ to $B$, such a subuniverse is a \emph{subalgebra} of $\mathbb{A}$, which we denote $\mathbb{B}\leq \mathbb{A}$.

If $\mathbb{A}_i$ is an indexed family of algebras of the same signature, the product $\prod_i \mathbb{A}_i$ of the family is the algebra whose universe is the Cartesian products of their universes $\prod_i A_i$ endowed with the basic operations which are coordinatewise products of the corresponding operations in $\mathbb{A}_i$. If $\mathbb{A}$ is an algebra, its $n$-th Cartesian power will be denoted $\mathbb{A}^n$.

An equivalence relation $\alpha$ on the universe $A$ of an algebra $\mathbb{A}$ is a \emph{congruence} of $\mathbb{A}$, if $\alpha \leq \mathbb{A}^2$, i.e. if $\alpha$ is preserved by all basic operations of $\mathbb{A}$. In that case, one can define the algebra $\mathbb{A}/\alpha$, the \emph{quotient of} $\mathbb{A}$ \emph{by} $ \alpha$, with the universe consisting of all equivalence classes (cosets) in $A/\alpha$ and whose basic operations are induced by the basic operations of $\mathbb{A}$. The $\alpha$-congruence class containing $a\in A$ will be denoted $a/\alpha$.

An algebra $\mathbb{A}$ is said to be \emph{simple} if its only congruences are the trivial, diagonal relation $0_\mathbb{A}=\{(a,a)\, \vert \, a\in A\}$ and the full relation $1_\mathbb{A}=\{ (a,b)\, \vert\, a,b\in A\}$. 

Any subalgebra of a Cartesian product of algebras $\mathbb{A}\leq \prod_i \mathbb{A}_{i\in I}$ is equipped with a family of congruences arising from projections on the product coordinates. We denote $\pi_i$ the congruence obtained by identifying the tuples in $A$ which have the same value in the $i$-th coordinate. Given any $J\subseteq I$, we can define a subalgebra of $\mathbb{A}$, $proj_J(\mathbb{A})$, which consists of the projections of all tuples in $A$ to the coordinates from $J$. If $\mathbb{A}\leq \prod_{i\in I} \mathbb{A}_i$ is such that $proj_i (\mathbb{A})=\mathbb{A}_i$, for every $i\in I$, we say that $\mathbb{A}$ is a \emph{subdirect product} and denote this fact $\mathbb{A}\leq_{sp} \prod_{i\in I} \mathbb{A}_i$.

If $\mathbb{A}$ and $\mathbb{B}$ are two algebras of the same signature, a mapping from $A$ to $B$ which preserves all basic operations is a \emph{homomorphism}. An \emph{isomorphism} is a bijective homomorphism between two algebras of the same signature.

Given an algebra $\mathbb{A}$, a \emph{term} is a syntactical object describing a composition of basic operations of $\mathbb{A}$. A \emph{term operation} $t^\mathbb{A}$ of $\mathbb{A}$ is the interpretation of the syntactical term $t(x_1,\ldots,x_m)$ as an $m$-ary operation on $A$, according to the formation tree of $t$.

A \emph{variety} is a class of algebras of the same signature, which is closed under the class operators of taking products, subalgebras, and homomorphic images (or, equivalently, under the formation of quotients by congruence relations.) The variety $\mathcal{V}(\mathbb{A})$ generated by the algebra $\mathbb{A}$ is the smallest variety containing $\mathbb{A}$. Birkhoff's theorem states (see \cite{Burris1981}) states that every variety is an equational class; that is, every variety $\mathcal{V}$ is uniquely determined by a set of identities (equalities of terms) $s\approx t$ so that $\mathbb{A}\in\mathcal{V}$ if and only if $\mathbb{A}\models s\approx t$, for every identity $s\approx t$ in the set.

An \emph{$n$-ary operation} on a set $A$ is a mapping
$f:A^n\rightarrow A$; the number $n$ is the \emph{arity} of $f$.  Let
$f$ be an $n$-ary operation on $A$ and let $k>0$. We write $f^{(k)}$
to denote the $n$-ary operation obtained by applying $f$ coordinatewise on
$A^k$. That is, we define the $n$-ary operation $f^{(k)}$ on $A^k$ by
\[
f^{(k)}(\mathbf a^1,\dots,\mathbf
a^n)=(f(a^1_1,\dots,a^n_1),\dots,f(a^1_k,\dots,a^n_k)),
\]
for $\mathbf a^1,\dots, \mathbf a^n\in A^k$.

The notion of \emph{polymorphism} plays the central role in the 
algebraic approach to the $\CSP$. 

\begin{definition}
  Given an $\Gamma$-structure $\mathbf{A}$, an $n$-ary
  \emph{polymorphism} of $\mathbf{A}$ is an $n$-ary operation $f$ on
  $A$ such that $f$ preserves the relations of $\mathbf A$. That is,
  if $\mathbf{a}^1,\dots,\mathbf{a}^n\in R$, for some $k$-ary relation
  $R$ in $\Gamma$, then $f^{(k)}(\mathbf a^1,\dots,\mathbf
  a^n)\in R$.  
\end{definition}

If a relational structure $\mathbf{A}$ is a core, one can construct a structure $\mathbf{A}'$ from $\mathbf{A}$ by adding, for each element $a\in A$, a unary constraint relation $\{a\}$. This enables us to further restrict the algebra of polymorphisms associated with the template; namely, if $f(x_1,\ldots,x_m)$ is an $m$-ary polymorphism of $\mathbf{A}'$, it is easy to see that
$f(a,a,\ldots,a)=a,$
for all $a\in A$. In addition to this, the constraint satisfaction problems with the templates $\mathbf{A}$ and $\mathbf{A}'$ are logspace equivalent. Therefore, we may assume that the algebra of polymorphisms associated to any CSP under consideration is \emph{idempotent}; i.e. all its basic operations $f$ satisfy the identity
$$f(x,x,\ldots,x)\approx x.$$

\begin{definition} A ternary operation $m:A^3\rightarrow A$ on a finite set is said to be \emph{Maltsev} if it satisfies the following algebraic identities
$$m(x,x,y)\approx m(y,x,x)\approx y.$$
\end{definition}

\begin{example}
A typical example of a constraint satisfaction problem over a finite Maltsev template is the problem of solving a system of linear equations in $n$ variables over a fixed finite field $K$, and let $S$ be its solution space, viewed as an $n$-ary relation on $K$. The operation $m(x,y,z)=x-y+z$ is a polymorphism of the relational structure $\mathbf{S}=(K; S)$. The converse is also true: namely, one can show that any $n$-ary relation on $K$, for $n\geq 1$, which has $m(x,y,z)$ as its polymorphism is a solution of some system of linear equations over $K$ in $n$ variables.
\end{example}

\subsection{Simple Idempotent Algebras in Maltsev Varieties}

Let $\mathbb{A}$ be an algebra. We say that $0\in A$ is an \emph{absorbing element} for $\mathbb{A}$ if, for every $(k+1)$-ary term operation $t(x,\bar{y})$, such that $t^\mathbb{A}$ depends on the variable $x$, the following holds for every $\bar{a}\in A^k$:
$$t^\mathbb{A}(0,\bar{a})=0.$$
We remark here that the property of being an absorbing element is stronger than the requirement that $\{0\}$ be an absorbing subuniverse of $\mathbb{A}$.

Given any finite power of an algebra $\mathbb{A}$, say $\mathbb{A}^n$, for $n\geq 2$, and any $n$ congruences $\theta_1,\theta_2,\ldots,\theta_n\in Con(\mathbb{A})$, the binary relation defined on $A^n$ by
$$((a_1,a_2,\ldots,a_n),(b_1,b_2,\ldots,b_n))\in \theta_1\times\theta_2\times\ldots\times\theta_n$$
if and only if $(a_i,b_i)\in \theta_i$, for all $i=1,\ldots,n$, is a congruence on $\mathbb{A}^n$. Therefore,
$$Con(\mathbb{A}_1)\times Con(\mathbb{A}_2)\times\ldots\times Con(\mathbb{A}_n)\subseteq Con(\mathbb{A}^n).$$
We say that a simple algebra $\mathbb{A}$ is \emph{congruence skew-free} if the equality holds, i.e. if 
$$Con(\mathbb{A}^n)\cong \mathbf{2}^n,$$
for every $n\geq 1$, where $\mathbf{2}$ is a two-element lattice.

In order to construct our algorithm, we will need to have complete understanding of simple idempotent Maltsev algebras. Our first step in this direction is the characterization of all idempotent simple algebras.

\begin{theorem} (K. Kearnes, \cite{Kearnes}) If $\mathbb{A}$ is an idempotent simple algebra, then exactly one of the following conditions is true:

\begin{enumerate}
\item $\mathbb{A}$ has a unique absorbing element.
\item $\mathbb{A}$ is Abelian.
\item $\mathbb{A}$ is congruence skew-free.
\end{enumerate}
\end{theorem}

It is easy to see that any algebra $\mathbb{A}$ which has an absorbing element cannot be Maltsev. Namely, suppose $a\in A$ is an absorbing element and let $m(x,y,z)$ be  a Maltsev polymorphism of $\mathbb{A}$. Then, for any $b\in A$, $b\neq a$, 
$$m(a,a,b)=a,$$
which is a contradiction. Therefore, the only possibilities for simple algebras in Maltsev varieties are Abelian algebras and the congruence skew-free ones.

Abelian idempotent simple algebras have an even more specific universal algebraic characterization:
 
\begin{theorem} (M. Valeriote, \cite{Valeriote}) Every simple Abelian  algebra is strictly simple, i.e. it contains no proper nontrivial subalgebras.
\end{theorem}

The structure of strictly simple idempotent Abelian algebras is well understood  (see e.g. \cite{szendrei}) and is closely related to modules over rings of matrices whose entries come from some fixed finite field:

A finite idempotent Abelian algebra $\mathbb{A}$ is strictly simple if and only if there exist a finite field $K$ and a finite-dimensional vector space $V$ over $K$ such that $\mathbb{A}$ is term equivalent to the algebra
$$(V; x-y+z, \{\lambda x+ (1_K-\lambda )y \, \vert  \, \lambda \in K\})$$
where $+$ is the addition of vectors, $1_K$ is the multiplicative identity of the field $K$, and $\lambda x$ is the scalar multiplication by $\lambda \in K$ in $V$. 

Simple congruence skew-free Maltsev algebras, unfortunately, do not allow for such a nice representation theorem. However, such algebras possess a rich polymorphic structure which, as we will see later, places any CSP parametrized by such an algebra in logarithmic space.

By a $k$-ary \emph{polynomial} of an algebra $\mathbb{A}$, we mean any function $p: A^k\rightarrow A$ which is obtained from a $(k+l)$-ary polymorphism $f(x_1,\ldots x_k,x_{k+1},\ldots,x_{k+l})$ and $a_1,\ldots, a_l\in A$, so that 
$$p(x_1,\ldots,x_k)=f^{\mathbb{A}}(x_1,\ldots,x_k,a_1,\ldots, a_l).$$

\begin{definition} An algebra $\mathbb{A}$, not necessarily idempotent, is said to be \emph{functionally complete}, if every finitary function on $\mathbb{A}$ is expressible by a polynomial.
\end{definition}

\begin{proposition} (H. Werner; for a proof see \cite{Burris1981}) If $\mathbb{A}$ is a finite algebra such that $\mathcal{V}(\mathbb{A})$ is Maltsev, then $\mathbb{A}$ is congruence skew-free if, and only if, $\mathbb{A}$ is functionally complete.
\end{proposition}

Finally, we state some facts about subdirect products of Maltsev algebras, which will be needed later.

The first fact concerns the connectivity in subdirect products of simple Maltsev algebras. For the proof, see e.g. \cite{Burris1981}

\begin{theorem} \label{maltsev} Let $\mathbb{A}_1,\ldots,\mathbb{A}_n$ be simple algebras in a Maltsev variety. If
$$\mathbb{B}\leq_{sp} \mathbb{A}_1\times\ldots\times\mathbb{A}_n$$
is a subdirect product, then
$$\mathbb{B}\cong \mathbb{A}_{i_1}\times\ldots\times\mathbb{A}_{i_k}$$
for some $\{i_1,\ldots,i_k\}\subseteq \{1,\ldots,n\}$. 

In particular, if $\mathbb{A}$ and $\mathbb{B}$ are two Maltsev algebras then any subdirect product
$$\mathbb{C}\leq_{sp} \mathbb{A}\times\mathbb{B}$$ 
is either the direct product or the graph of an isomorphism $f:\mathbb{A}\rightarrow\mathbb{B}$.
\end{theorem}

The existence of a Maltsev operation implies the following property on any subdirect product, which we will refer to as the \emph{rectangularity property}:

\begin{proposition} Let $C\leq_{sp} \mathbb{A}\times\mathbb{B}$, where $\mathbb{A}$ and $\mathbb{B}$ are algebras in a Maltsev variety. Then, the following holds: if $(a,b), (a,b'), (a',b')\in C$, then $(a',b)\in C$.
\end{proposition}

\begin{proof} Let $m(x,y,z)$ be a Maltsev polymorphism on both $\mathbb{A}$ and $\mathbb{B}$. Then,
$$(a',b)=(m(a,a,a'), m(b,b',b'))\in C.$$
\end{proof}

Subdirect products of a pair of (not necessarily Maltsev) algebras give rise to pairs of congruences which will be used to provide finer structural analysis of subdirect products of pairs of algebras.

\begin{proposition}\label{absorb} Let $\mathbb{R}\leq_{sp} \mathbb{A}\times \mathbb{B}$.
The binary relation $\alpha$ defined on $A$ by
$$(a,a')\in\alpha  \mbox{ if and only if  there exists $b\in B$ such that } (a,b),(a,b')\in C$$
is a congruence of $\mathbb{A}$. The analogous statement is true of the dual relation $\beta$ defined on $B$.

\end{proposition}

We will refer to the congruences $\alpha$ and $\beta$, defined as in Proposition \ref{absorb} , as the \emph{linkedness} congruences on $\mathbb{A}$ and $\mathbb{B}$ induced by $\mathbb{C}$. We say that $\mathbb{A}$ and $\mathbb{B}$ are \emph{linked} if $\alpha=1_\mathbb{A}$ and $\beta=1_\mathbb{B}$ or, equivalently, if $\pi_1\vee\pi_2=1_\mathbb{C}$. If $\alpha=0_\mathbb{A}$ and $\beta=0_\mathbb{B}$, the subdirect product is the graph of an isomorphism between the algebras $\mathbb{A}$ and $\mathbb{B}$.

Based on the definition of linkedness congruences (Proposition \ref{absorb}), the following statement has an elementary proof, using the rectangularity property:

\begin{corollary} \label{full} Let $C\leq_{sp} \mathbb{A}\times \mathbb{B}$, where $\mathbb{A}$ and $\mathbb{B}$ are algebras in a Maltsev variety, and let $\alpha$ and $\beta$ be the linkedness congruences associated with this subdirect product, on $\mathbb{A}$ and $\mathbb{B}$, respectively, as defined in Proposition \ref{absorb}. Then, if $C_1$ and $C_2$ are blocks of $\alpha$ and $\beta$, respectively, which are connected, 
$$C_1\times C_2\leq \mathbb{A}\times\mathbb{B}.$$
\end{corollary}

\section{Reduction to binary relations}\label{Binary}

In this section, we outline the reduction of an arbitrary instance with a sufficient degree of consistency to a binary one. The construction is due to L. Barto and M. Kozik and we largely adhere to their exposition in \cite{b-k1}.

An instance is said to be \emph{syntactically simple} if it satisfies the following conditions:

\begin{itemize}
\item every constraint is binary and it its scope is a pair of distinct variables $(x,y)$.
\item for every pair of distinct variables $x,y$, there is at most one costraint $R_{x,y}$ with the scope $(x,y)$.
\item if $(x,y)$ is the scope of $R_{x,y}$, then $(y,x)$ is the scope of the constraint $R_{y,x}=\{(b,a) \, \vert\, (a,b)\in R_{x,y}\}$ (\emph{symmetry of constraints}).
\end{itemize}

Given the Maltsev algebra $\mathbb{A}$ parametrizing the instance $\mathcal{I}$, such that the maximal arity of a relation in $\mathcal{I}$ is $p$, we run the algorithm verifying the $(2\lceil\frac{p}{2}\rceil,3\lceil\frac{p}{2}\rceil)$-consistency on $\mathcal{I}$. If the algorithm terminates in failure, we output ``$\mathcal{I}$ has no solution." If the algorithm terminates successfully, we output a new, syntactically simple instance $\mathcal{I}'$ in the following way:

\begin{itemize} 
\item The instance is parametrized by $\mathbb{A}^{\lceil \frac{p}{2}\rceil}$, which is a Maltsev algebra. 
\item For every $\lceil\frac{p}{2}\rceil$-tuple of variables in $\mathcal{I}$, we introduce a new variable in $\mathcal{I}'$ and, if $x=(x_1,\ldots,x_{\lceil\frac{p}{2}\rceil})$ and $y=(y_1,\ldots,y_{\lceil\frac{p}{2}\rceil})$ with $x\neq y$, we introduce a constraint 
\begin{multline*}
R_{x,y}=\{((a_1,\ldots,a_{\lceil\frac{p}{2}\rceil}),(b_1,\ldots,b_{\lceil\frac{p}{2}\rceil}))\,\vert \\ (a_1,\ldots,a_{\lceil\frac{p}{2}\rceil},b_1,\ldots,b_{\lceil\frac{p}{2}\rceil}) \mbox{ admit a consistent $p$-assignment of values }\}.
\end{multline*}
\end{itemize}

The binary instance $I'$ constructed in this way will have a solution if, and only if, the instance $I$ has a solution.

\begin{definition} A \emph{step}  in an instance $\mathcal{I}$ is a pair of variables which is the scope of a constraint in $\mathcal{I}$. A \emph{path-pattern} from $x$ to $y$ in $\mathcal{I}$ is a sequence of steps such that every two steps correspond to distinct binary constraints and which identifies each step's end variable with the next step's start variable. A \emph{subpattern} of a path-pattern is a path-pattern defined by a substring of the sequence of steps. We say that a path-pattern is a \emph{cycle} based  at $x$ if both its start and end variable are $x$.
\end{definition}

\begin{definition} Let
$$p=(x_1,x_2,\ldots,x_k)$$
be a path-pattern. A \emph{realization} of $p$ is a $k$-tuple $(a_1,\ldots,a_k)\in \mathbb{S}_{x_1}\times\ldots\times \mathbb{S}_{x_k}$ such that $(a_i,a_j)$ satisfies the binary constraint associated with the $(x_i,x_j)$-step. If $p$ is a path-pattern with the start variable $x_i$ and $A\subseteq S_{x_i}$, we denote $A+p$ the set of the end elements of all realizations of $p$ whose first element is in $A$. $-p$ will denote the inverse pattern of $p$, i.e. the pattern obtained by reversing the traversal of the pattern $p$. In that case, we define $A-p = A+(-p)$.
\end{definition}

\section{Cyclic Maltsev Constraint Satisfaction Problems}

We will say that a Maltsev constraint satisfaction problem is \emph{cyclic} if its domains are isomorphic simple algebras, all of the constraints are binary, and it is is 1-consistent. As we have seen, in the context of Maltsev algebras, this implies that each constraint relation between two domains $\mathbb{S}_{x_i}$ and $\mathbb{S}_{x_j}$ is either the graph of an isomorphism or a full direct product $\mathbb{S}_{x_i}\times \mathbb{S}_{x_j}$. 

Based on our earlier analysis of simple idempotent Maltsev algebras, there are two types of cyclic problems:

\begin{enumerate}
\item A system of linear equations in two variables over a finite field; 
\item A binary CSP over a simple, functionally complete Maltsev algebra.
\end{enumerate}

For each cyclic CSP $\mathcal{I}$ over a family of simple Maltsev algebras, all isomorphic to some fixed simple Maltsev algebra  $\mathbb{A}$, we can define the accompanying undirected \emph{instance graph} $G(\mathcal{I})$ in the following way: the vertices of the graph are all domains $\mathbb{S}_x$ of $\mathcal{I}$ and two vertices $\mathbb{S}_x$ and $\mathbb{S}_y$ have an edge between them if, and only if, the binary constraint relation $R_{x,y}$ is the graph of an isomorphism. We can compute the connected components of this graph in logspace, using Reingold's algorithm (\cite{Reingold05}). 

It is not difficult to see that, in order to solve such a CSP, we need to be able to solve it independently in each connected component of $G(\mathcal{I})$. Fix a connected component $C$ of  $G(\mathcal{I})$, and assume that the domain indices appearing in the definition of $G(\mathcal{I})$ are $C=\{x_{i_1},\ldots,x_{i_r}\}$. We restrict the instance $\mathcal{I}$ to this set of indices to obtain $\mathcal{I}_C$ and remove the constraints between any $\mathbb{S}_{x_{i_j}}$ and $\mathbb{S}_{x_{i_k}}$ such that there is no edge between $x_{i_j}$ and $x_{i_k}$ in $G(\mathcal{I})$. 

We obtain the following proposition, whose proof follows directly from the definitions of $G(\mathcal{I})$ and $\mathcal{I}_C$:

\begin{proposition} Suppose $a\neq b$ with $a,b\in \mathbb{S}_{x{i_j}}$ and $x_{i_j}\in C$. Then, no solution to $\mathcal{I}_C$ (or $\mathcal{I}$) $f$ can satisfy $f(x_{i_j})=a$ and $f(x_{i_k})=c$, for $c\in\mathbb{S}_{x_{i_k}}$, $x_{i_k}\in C$, if there is a path from $c$ to $b$ in the undirected graph $G'$, whose vertices are elements of $\bigcup_{i=1,\ldots, r}\mathbb{S}_{x_i}$ and whose edge relation is defined by the binary constraints in $\mathcal{I}_C$.
\end{proposition}

\begin{proof} For the sake of contradiction, we assume that such a solution $f$ exists. Let $p$ be the undirected path in $C$ connecting $x_{i_j}$ to $x_{i_k}$, such that one realization of that path pattern connects $b$ to $c$. Consider the path $p': a=f(x{i_j}),\ldots, f(x_{i_k})=c$, obtained by projecting the solution $f$ to all the domains whose indices appear in $p$. However, the definitions of $G(\mathcal{I})$ and $G'$ imply that two paths originating in the same $\mathbb{S}_{x_{i_j}}$ with different starting points cannot have the same endpoint in any $\mathbb{S}_{x_{i_k}}$, assuming $x_{i_j},x_{i_k}\in C$. Contradiction. Therefore, such a solution cannot exist.\end{proof}

As a consequence of this proposition, we see that the following must hold

\begin{corollary} The binary instance $\mathcal{I}_C$ has a solution if, and only if, there does not exist $c\in\mathbb{S}_{x{i_k}}$, $x_{i_k}\in C$, which is reachable from both $a$ and $b$, for some $a,b,a\neq b$ and $a,b\in\mathbb{S}_{x{i_j}}$ in the undirected graph $G'$ whose vertices are elements of $\bigcup_{i=1,\ldots, r}\mathbb{S}_{x_i}$ and whose edge relation is defined by the binary constraints in $\mathcal{I}_C$.
\end{corollary}

Therefore, a cyclic CSP over a simple Maltsev domain can be solved in deterministic logspace, using Reingold's algorithm for reachability in undirected graphs.

\section{An algorithm for Maltsev CSPs }

In this section, we will construct a polynomial time algorithm for solving constraint satisfaction problems over a finite Maltsev template.

The underpinning of the algorithm will be the new concept of Maltsev consistency which pre-processes the instance by removing cyclic CSP subinstances which do not contain any solutions. The consistency check will produce a polynomial family of subinstances in which solutions must lie. We can view such subinstances as ``passive"; the reductions performed on the current active instance are implicitly performed on the passive subinstances in such a way that Maltsev consistency is preserved. A reader familiar with Bulatov-Dalmau algorithm will notice a smiliarity between the passive subinstances and the notion of compact representation (i.e. signature) in that algorithm. This similarity is not coincidental; the tuples in the compact representation of a relation will belong to passive subinstances.

\subsection{Maltsev consistency}

We define inductively, the Maltsev-consistency checking agorithm $\mathcal{M}_k$, for all CSP instances $\mathcal{I}$ such that $\max_{i}|\mathbb{S}_{x_i}|\leq k$.

Let $A\leq \mathbb{S}_{x_i}$. For any $x_j\in V$, $x_j\neq x_i$, we define $R^+_{x_i,x_j}(A)=\{b\in \mathbb{S}_{x_j} \, : \, \exists a\in A, \, (a,b)\in R_{x_i,x_j}\}$.  Clearly, $R^+_{x_i,x_j}(A)$ is a subuniverse of $\mathbb{S}_{x_j}$.

We will define Maltsev consistency in such a way, that all subinstances which are CSPs over isomorphic simple subuniverses and which do not have a solution, are excluded by the consistency check. In addition, such subinstances which cannot be extended to a larger (2,3)-consistent subinstance are excluded.

We start by defining the list $\mathcal{A}$ of all pairs $(A,\theta_A)$, where $A$ is a subuniverse of some $\mathbb{S}_{x_i}$ and $\theta_A$ is a maximal congruence of $A$.We can arrange the list $\mathcal{A}$ in such a way that, if  $(A,\theta_A)$ and $(A',\theta_{A'})$ are two elements of the list and $A'$ is contained in a $\theta_A$-block of $A$, then $(A,\theta_A)$ appears in the list $\mathcal{A}$ before $(A',\theta_{A'})$.

We are now ready to state the procedure which enforces Maltsev consistency

\begin{enumerate}
\item For the next pair $(A,\theta_A)$ in the list $\mathcal{A}$, form the $(A,\theta_A)$-\emph{test instance} in the following way: suppose $A\leq \mathbb{S}_{x_i}$, for some $1\leq i\leq n$. For $x_j\neq x_i$, if there exists a congruence $\alpha_{x_j}$ on $R_{x_i,x_j}(A)$, such that, if $B_1$ and $B_2$ are two distinct $\theta_A$-blocks and $p$ a path pattern from $x_i$ to $x_j$ such that $R^+_{x_i,x_j}(A)\cap (B_1+p)$ and $R^+_{x_i,x_j}(A)\cap (B_2+p)$ are containt in distinct blocks of $\alpha_{x_j}$, we will say that the variable $x_j$ is  \emph{relevant}. Therefore, for each relevant variable $x_j$,
$$\mathbb{S}/\alpha_{x_j}\cong A/\theta_A.$$
In fact, $\alpha_{x_j}$ is independent of the choice of the path pattern $p$, because of (2,3)-consistency. 

We define a \emph{strand} to be the set of those congruence blocks in each relevant domain which are linked to the same congruence block of $\theta_A$. The $(A,\theta_A)$-test instance will have as its domains the algebras $R^+_{x_i,x_j}(A)/\alpha_{x_j}$, for $j\neq i$, for all relevant variables $x_j$ and $A$ in the $j$-th coordinate, along with $\alpha_{x_i}=\theta_A$.  Since $\mathcal{I}$ is a (2,3)-consistent instance, for any pair of relevant variables $y,z$, distinct from $x$, the binary constraint $E_{y,z}$ induces a subdirect product on $R_{x,y}^+(A)$ and $R_{x,z}^+(A)$, so that the $(A,\theta_A)$-test instance is 1-consistent.

\item The $(A,\theta_A)$-test instance is a  cyclic CSP, and using Reingold's algorithm, we test whether blocks of $\theta_A$ appear in solutions or not; those which do not are removed.Furthermore, we check each solution strand for Maltsev consistency, using $\mathcal{M}_{k-1}$ and, if it fails, we remove the corresponding $\theta_A$-congruence block from the instance.

\item Enforce (2,3)-consistency.

\item Proceed to the next element in the list $(A',\theta_A')$, if there are any left.

\end{enumerate}

There are only polynomially many pairs in the list $\mathcal{M}$, so the algorithm for enforcing Maltsev consistency runs in polynomial time. In fact the number of test instances can be bounded above by $\mathcal{O}(n)$.

\begin{lemma} Let $\mathcal{I}$ be a syntactically simple binary instance and let $\mathcal{I'}$ be the instance produced by applying the Maltsev consistency algorithm to it. Then, the sets of solutions to $\mathcal{I}$ and $\mathcal{I'}$ coincide.
\end{lemma}

\begin{proof} If there exists a solution $f$ to $\mathcal{I}$ whose projection to the $x$-coordinate is in $A\leq \mathbb{S}_x$, then, its restriction to relevant variables is also a solution of the $(A,\alpha_A)$-test instance, viewed as a subinstance of $\mathcal{I}$. If Maltsev consistency test fails on an $\alpha_A$-block, then there cannot be any solutions $f$ projecting into that block in their $x$-coordinate.

Also, the solution projecting into a $\alpha_A$-block $B$ in its x-coordinate will lie in its entirety in the subinstance induced by $B$, so this subinstance must be (2,3)-consistent.
\end{proof}

The following lemma, which is an immediate consequence of Corollary \ref{full}, will play a crucial role in the construction of the algorithm for solving CSPs over Maltsev templates:

\begin{lemma} \label{rect} Let $A\leq \mathbb{S}_{x_i}$, for some $1\leq i\leq n$. Suppose $x_j$ ($j>i$) is a non-relevant variable for the $(A,\theta_A)$-test instance. Then the induced subdirect product 
$$C\leq_{sp} A/\theta_A\times R^+_{i,j}(A)$$
is a full direct product.
\end{lemma}

Assuming the instance passes the Maltsev consistency test, we will also obtain a polynomial family of subinstances $\mathcal{P}$, henceforth referred to as \emph{passive subinstances}, which are those subinstances examined in Step 3 of the algorithm, and which are (2,3)-consistent, after the Maltsev consistency check has terminated successfully. The instances in $\mathcal{P}$ are updated, i.e. modified, by the reductions induced by the ones that will be performed on the active instance $\mathcal{I}$.

\subsection{Reduction to smaller subinstances}\label{reduction}

In this subsection, we show that any 1-consistent, Maltsev-consistent instance $\mathcal{I}$ can be reduced to a proper subinstance $\mathcal{I}'$, which satisfies the same consistency properties, in such a way that, at the last step of the reduction, one of the domains has been reduced to a singleton. We assume that the set of variables $V$ has been ordered as $\{x_1,x_2,\ldots,x_n\}$. Our algorithm will proceed in the following manner: we first reduce $\mathbb{S}_{x_1}$ to a singleton while maintaining the required consistency conditions, in order to obtain a smaller subinstance $\mathcal{I}'$. Following this, an analogous sequence of reduction is carried out on $\mathbb{S}_{x_2}$, etc, until all domains have been reduced to single elements, which yields a solution to the original CSP.

We will say that a 1-consistent instance $\mathcal{I}$ with the set of variables $V=\{x_1,x_2,\ldots,x_n\}$ is $i$-\emph{processed} if, the domains $\mathbb{S}_{x_1},\ldots,\mathbb{S}_{x_i}$ have been reduced to single elements.

\begin{proposition} Let $\mathcal{I}$ be a 1-consistent, Maltsev-consistent instance over a finite Maltsev template, which is $i$-processed and in $\mathcal{A}$. Then, there is a passive subinstance $\mathcal{I}'\in\mathcal{A}$ which is 1-consistent, Maltsev-consistent and $(i+1)$-processed.
\end{proposition}

\begin{proof} If $|\mathbb{S}_{x_{i+1}}|=1$, the statement holds trivially. Therefore, we may assume that $\mathbb{S}_{x_{i+1}}$ is non-trivial.

Let $\alpha_{x_{i+1}}$ be a maximal congruence on $\mathbb{S}_{x_{i+1}}$. The strands of the test instance $\mathbb{S}_{x_{i+1}}/\alpha_{x_{i+1}}$ are in $\mathcal{A}$ but what remains to be shown, first of all, is that no non-relevant variable in the test instance becomes relevant when taking the intersection with $\mathcal{I}$. By Lemma \ref{rect}, if $j>i+1$ and $j$ was non-relevant in the test instance, the subdirect product
$$C'\leq_{sp} \mathbb{S}_{x_{i+1}}\times S_{x_j}$$
is the full direct product, since $S_{x_j}\leq R^+_{i+1,j}(S_{x_{i+1}})$.
Furthermore, suppose that the processed variables $x_1,\ldots,x_i$ have been reduced to $a_1,\ldots,a_i$. Since $a_k$, ($k=1,\ldots,i$), is linked to all vertices of all domains $\mathbb{S}_l$ ($l=i+1,\ldots,n$), $x_1,\ldots, x_i$ are non-relevant variables in all the remaining passive subinstances in $\mathcal{P}$.

Next, we can take any block $B$ of $S_{x_{i+1}}/\alpha_{x_{i+1}}$, replace $S_{x_{i+1}}$ with $B$ and continue the reduction in the coordinate $x_{i+1}$ until we get an $(i+1)$-processed instance, which is 1-consistent.

\end{proof}

\subsection{Algorithm}

In conclusion, we summarize the polynomial algorithm which, on input an instance $\mathcal{I}$ of $\CSP (\mathbf{A})$, where the algebra $\mathbb{A}$ is Maltsev, determines whether the instance has a solution and outputs one if it exists:

\begin{algorithm}[H]
   \caption{Algorithm for solving CSP instances for a Maltsev algebra $\mathbb{A}$}
    \begin{algorithmic}[1]
    
        \State if $p$ is the maximum arity of the operation of the constraint language $\mathbf{A}$, run the $(2\lceil\frac{p}{2}\rceil,3\lceil\frac{p}{2}\rceil)$-consistency algorithm on $\mathcal{I}$
                   \If { no $(2\lceil\frac{p}{2}\rceil,3\lceil\frac{p}{2}\rceil)$-consistency}
                        \State  output \emph{`no solution'}
                   \Else
                         \State verify Maltsev consistency on $\mathcal{I}$
                                  \If { no Maltsev consistency}
                                       \State output \emph{ `no solution'}
                                 \Else
                                        \State Enforce (2,3)-consistency on $\mathcal{I}$
                                        \State Generate the system $\mathcal{P}$ of polynomially many passive subinstances 
                                        \State output a solution using reductions via maximal congruences.
                                  \EndIf
                    \EndIf
\end{algorithmic}
\end{algorithm}

\section{Conclusion}

We have presented an algorithm for solving constraint satisfaction problems over finite templates with Maltsev polymorphisms, which is based on a new type of consistency check,the Maltsev consistency, which solves localized subproblems over simple algebras in the variety generated by the parametrizing algebra $\mathbb{A}$,appearing in the computation tree and, effectively, removes the unsuccessful branches leading to no solutions.

One of the questions which have inspired the work presented in this article was the question which extensions of first-order logic, which are thought of as potential candidates for capturing polynomial time on finite relational structures, have the capability of expressing the solvability of CSPs with Maltsev templates. One such candidate logic is the Choiceless Polynomial Time with Counting (for more details, see \cite{gradel2015polynomial} ) In that article, the following problem was posed:

\begin{openproblem} Are all CSPs with Maltsev constraints expressible in the Choiceless Polynomial Time with Counting?
\end{openproblem}

The affirmative answer to this question would imply that the isomorphism of graphs with bounded colour size is expressible in this logic, via a reduction given in \cite{berkholz2015limitations}. We speculate that the answer to this problem is affirmative; namely, the algorithm is based on solving a polynomial number of ``atomic" problems which are either definable in Symmetric Datalog, or are cyclic systems of linear equations over a finite field. Since Choiceless Polynomial Time with Counting subsumes the expressive power of LFP and cyclic systems of linear equations are expressible in CPT+C (\cite{pakusa2015}), there is strong evidence that this indeed should be the case.

\bibliography{digraph_reduction} 
\bibliographystyle{siam}

\end{document}